\documentclass[11pt,letterpaper]{article}
\usepackage{amsmath}
\usepackage{authblk}
\usepackage{amsthm}
\usepackage{upgreek}
\usepackage{amsfonts}
\usepackage{fullpage}
\usepackage{float}
\usepackage{tikz}
\usepackage{amssymb}
\usepackage{mathtools}
\usepackage{mathrsfs}  
\usepackage{varwidth}
\usepackage{sidecap}
\usepackage{url}
\usepackage{caption}
\usepackage{enumitem}
\usepackage[linesnumbered,ruled,vlined]{algorithm2e}
\usepackage{cite}
\usepackage[none]{hyphenat}
\usepackage{breqn}
\usepackage[section]{placeins}
\usepackage{float}
\usepackage{lipsum}
\usepackage{breqn}
\usepackage{setspace}
\usepackage{array}
\usepackage{ragged2e}
\usetikzlibrary{shapes,arrows}
\usepackage[left=1in,right=1in,top=1in,bottom=1in]{geometry}
\usepackage{url}
\newtheorem{theorem}{Theorem}

\newtheorem{definition}{Definition}[section]
\newtheorem{lemma}{Lemma}[section]

\newtheorem{claim}{Claim}[lemma]
\newtheorem{theoremdup}{Theorem}

\usepackage[utf8]{inputenc}

\newcommand{\norm}[1]{\left\lVert #1 \right\rVert}
\RestyleAlgo{boxruled}
\providecommand{\keywords}[1]{\textbf{{Keywords.}} #1}
\title{\normalsize \bf{A polynomial time 12-approximation algorithm for the restricted Santa Claus problem }}
\author[1]{S. Anil Kumar \thanks{Indian Institute of Technology Madras, cs18d001@cse.iitm.ac.in}}
\author[2]{N.S. Narayanaswamy \thanks{Indian Institute of Technology Madras, swamy@cse.iitm.ac.in}}
\affil[1,2]{Indian Institute of Technology Madras}
\date{}
\allowdisplaybreaks
\begin{document}
\maketitle
\vspace{-15pt}
\begin{abstract}
The Santa Claus problem is a well known $NP$-hard problem in the class of resource allocation problems. In this problem, we are given a set $\mathcal{J}$ of jobs to be allocated among a set of machines $\mathcal{M}$ in such a way that machine $i$ on allocating job $j$ gets size $p_{ij}$.  We are required to allocate the jobs in $\mathcal{J}$ among the machines in $\mathcal{M}$ so as to maximise the minimum total size on all machines. Even though, this general case of the problem has unbounded integrality gap \cite{bansal}, the restricted case of the problem suggested by Bansal and Sviridenko \cite{bansal} is known to have constant integrality gap\cite{AFS}\cite{AFS2}. The best known polynomial time approximation algorithm for solving this restricted case has an approximation ratio of 13. 

In this paper, we consider the restricted case of the problem and improve the current best approximation ratio by presenting a polynomial time 12-approximation algorithm  using linear programming and semi-definite programming. Our algorithm starts by solving the configuration LP and uses the optimum value to get a 12-gap instance.  This is then followed by the well-known clustering technique of Bansal and Sviridenko\cite{bansal}. We then apply the analysis of Asadpour \textit{et al.} \cite{AFS,AFS2} to show that the clustered instance has an integer solution which is at least $\frac{1}{6}$ times the best possible value, which was computed by solving the configuration LP.  To find this solution, we formulate a problem called the Extended Assignment Problem, and formulate it as an LP. We then, show that the associated polytope is integral and gives us an fractional solution of value at least $\frac{1}{6}$ times the optimum.  
From this solution we find a solution to  a new quadratic program that we introduce to select one machine from each cluster, and then we show that the resulting instance has an Assignment LP fractional solution of value at least $\frac{1}{6}$ times the optimum.  We then use the well known rounding technique due to Bezakova and Dani \cite{bezakova} on the 12-gap instance to get our 12-approximate solution. 

\end{abstract}
\section{Introduction}
Resource Allocation problems are fundamental to computing. They provide a vast class of challenging algorithmic questions in many computational settings.  
A wide class of these problems are max-min fair allocation problems, which focus on distributing the available resources to the underlying entities in the fairest way possible. Santa Claus problem is one of such problems and was proposed by Bansal and Sviridenko\cite{bansal}. A  Santa Claus instance is a triplet where $\mathcal{M}$ is a set of $m$ machines, $\mathcal{J}$ is a set of $n$ jobs to be allocated among the machines in $\mathcal{M}$ and $\mathcal{P}$ is a table of dimension $m \times n$ indexed by a pair of machine and job such $(i,j)^{th}$ entry constraints the size $p_{ij}$ exerted by job $j$ on machine $i$ and we are asked to find an allocation of the jobs in $\mathcal{J}$ among the machines in $\mathcal{M}$ in such a way that $\min\Big\{\sum\limits_{j \in \mathcal{J}}p_{ij} \Big|  i \in \mathcal{M}\Big\}$ is maximised.
   
The work due to Bansal and Sviridenko\cite{bansal} attempts to find the solutions to the problem by developing approximation algorithms using linear programming approach and establishes one negative result and two positive results. The negative result is that the general version of the problem problem has unbounded integrality gap that grows as $\Omega(\sqrt{m})$. Hence a lot of focus has been to understand the complexity of a restricted case of the problem which is known as the \emph{restricted Santa Claus problem} and is given in Figure \ref{restricted Santa Claus problem}.
\begin{figure}[H]
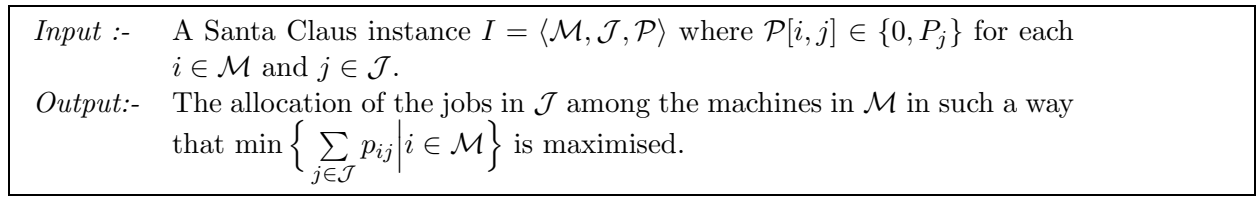

\fbox{\begin{minipage}{0.99\textwidth}
     \begin{tabular}{lp{12cm}}
 \textit{Input :-} & A Santa Claus instance $I=\langle \mathcal{M},\mathcal{J},\mathcal{P}\rangle$ where $\mathcal{P}[i,j] \in \{0,P_{j}\}$ for each $i \in \mathcal{M}$ and $j \in \mathcal{J}$.\\
 \textit{Output:-} & The allocation of the jobs in $\mathcal{J}$ among the machines in $\mathcal{M}$ in such a way that $\min\Big\{\sum\limits_{j \in \mathcal{J}}p_{ij} \Big|  i \in \mathcal{M}\Big\}$ is maximised. \\
\end{tabular}
     \end{minipage}}
\caption{The restricted Santa Claus problem \cite{bansal}}
\label{restricted Santa Claus problem}
\end{figure}

A natural interpretation of the restricted case of the problem is that each job is interested in only some machines for the allocation and has identical size on all such machines. The first positive result in the work due to Bansal and Sviridenko \cite{bansal} gives a $O\Big(\frac{\ln{n}}{\ln{\ln{n}}}\Big)$ randomised approximation algorithm for solving the restricted Santa Claus problem using a machine job clustering technique and randomised rounding. The second positive result improves the first result by presenting a $O\Big(\frac{\ln{\ln{m}}}{\ln{\ln{\ln{m}}}}\Big)$ approximation algorithm. The paper also specifies a combinatorial conjecture which if true establishes the constant integrality gap of the restricted version.

Later Asadpour \textit{et al.}\cite{AFS} established two positive results in connection with the restricted Santa Claus problem. They showed that restricted Santa Claus problem has a $\frac{1}{3}$ integrality gap by proving the conjecture given by Bansal and Sviridenko using the Haxell's condition for hypergraph perfect matching\cite{haxel}. They also established a $5$-approximation algorithm for the restricted case by treating the problem as a matching problem in hypergraphs and finding suitable allocations by introducing the notion of alternating tree and this algorithm was later refined to give a $4$-approximate solution\cite{AFS2}. Even though this method provides approximate solutions nearer to the best theoretically known optimum, it lacks due to the exponential time complexity involved in the underlying computation. 

A work that is related with the original work due to Bansal and Sviridenko was due to Feige\cite{feige} which shows that the $O\Big(\frac{\ln{\ln{m}}}{\ln{\ln{\ln{m}}}}\Big)$ algorithm in fact approximates the optimum value to within a constant factor. 
Even though the work due to Feige established the existence of constant factor approximation algorithms for the restricted Santa Claus problem, it does not give constructive ideas to develop such constant approximation factor approximation algorithms which run in polynomial time.  Later in 2017,  Chidambaram \textit{et al.}\cite{annamalai1} formulated two approximation algorithms with constant approximation factor and running in polynomial time for solving the restricted case. The first results presents a 36-approximation which uses the machine job clustering technique given in the work due to Bansal and Svridenko along with the alternating tree technique given in the work due to Asadpour \textit{et al.}. The second algorithm is a 13-approximation algorithm which is a combinatorial algorithm. Both the algorithm uses techniques such as lazy updates and greedy choice to limit the running time to within polynomial factor.   

A further related work on Santa Claus problem due to Chakrabarty \textit{et.al.}\cite{chakrabarty} showed the existence of a $O(n^{\epsilon})$ approximation algorithm for the problem which runs in time $n^{O\big(\frac{1}{\epsilon}\big)}$ using a linear program with integrality gap $\Omega(\sqrt{m})$. They also investigated a special case of the problem where each job has non zero size for at most two machines and proved that even this special case is $NP$-hard to approximate within a factor less than 2.  Another related work due to Woeginger\cite{woeginger} investigated an even more restricted version namely the uniform case where the size of a job on all the machines are the same and they derived a PTAS for this case. 
\subsection{Our Results}
In this paper, we focus on solving the restricted Santa Claus problem given in figure \ref{restricted Santa Claus problem}. This work is motivated by the existence of large gap between the best theoretical integrality gap of $3$ for the restricted Santa Claus problem as evident from result due to Asadpour \textit{et al.}\cite{AFS,AFS2} and the approximation ratio of $13$ given by the best known polynomial time algorithm due to Chidambaram \textit{et al.}\cite{annamalai1}. We narrow this gap by designing a polynomial time approximation algorithm having an approximation guarantee of 12. More precisely, we prove the following theorem.
\begin{theorem}
\label{thm1}
There exists a polynomial time 12-approximation algorithm for the restricted Santa Claus problem.
\end{theorem} 
Our methods are primarily based on Linear Programming with several novel  techniques.
\begin{enumerate}
\item We use the machine job clustering technique due to Bansal and Sviridenko\cite{bansal} as the basic platform to design our algorithm. We show that a certain way of machine job clustering ensures the existence an integer solution which assigns small job configurations of size $\frac{T}{6}$ to exactly one machine in each super machine, if we start with a 12-gap instance (Theorem \ref{thm2}). We establish such a machine job clustering by grouping the machines in the input Santa Claus instance into upper class machines and middle class machines. This approach of grouping the machines ensures that we can form variants of Configuration LP involving all the machines and the small jobs in the given instance.  We then show that the clustered instance has an integer solution of size $\frac{T}{6}$, by applying the techniques of Asadpour \textit{et.al} \cite{AFS} \cite{AFS2}. 
We then show that this integer solution is a feasible solution for a configuration LP on the clustered instance with bundles of size $\frac{T}{6}$.  We then show that this implies that a bilinear program which we call the Extended Assignment Program has a fractional solution of value at least $\frac{T}{6}$.
\item Using the fractional solution to the Extended Assignment Program, we obtain a deterministic polynomial time algorithm to choose a single machine from each super machine, which can be assigned small job configurations. In particular, our algorithm finds an integer $12$-approximate solution that allocates small job configurations to a list of machines containing exactly one machine from each super machine and the middle class machines.  This is achieved by applying the rounding method due to Bezakova and Dani\cite{bezakova} to obtain the desired 12-approximate solution.
\end{enumerate}
The following are the notable attributes and consequences of our results.
\begin{enumerate}
\item The algorithm is conceptually simple and is yet another witness to the effectiveness of the configuration LP to obtained improved approximation ratios.
\item The algorithm neither uses any alternating tree technique\cite{AFS,AFS2} nor apply any greedy techniques such as Lazy Local Search and Greedy Players\cite{annamalai1} to find the 12-approximate solution. We use the alternating tree method due to Asadpour \textit{et.al}\cite{AFS,AFS2} only to prove that there exists an integer solution that allocates small jobs among the middle class machines and exactly one machine from each super machine in such a way that each of these super machines gets a total size of at least $\frac{T}{6}$ (Theorem \ref{thm2}). 
\end{enumerate}
The algorithm does not however have the features of the combinatorial algorithm designed by Chidambaram \textit{et.al} \cite{annamalai1}. 
\subsection{Organisation of the paper}
The remainder of the paper is organised as follows. In section 2, we discuss some building blocks of linear programming techniques which are used in Santa Claus problem. This includes a brief overview of the techniques used in the works by Bansal and Sviridenko \cite{bansal} and Asadpour \textit{et al.}\cite{AFS,AFS2}. Section 3 starts with an outline of our algorithm and detail each step as a separate subsection. In section 3.1, we present an outline of our algorithm. Section 3.2 explains how we can classify the set of machines in the given Santa Claus instance into upper class machines and middle class machines. Section 3.3 shows that we can compute a solution for the given instance in polynomial time, if there are no upper class machines within the given instance. In section 3.4, we deal with the case where the Santa Claus instance has one or upper class machines and is this section further divided into four sections. Section 3.4.1 explains how we can form the super machines out of the upper class machines and the composite machines. Section 3.4.2, we define the Modified Configuration LP and the Extended Configuration LP and prove Theorem \ref{thm2}. In section 3.4.3, we define the Extended Assignment Program and use Theorem \ref{thm2} to prove that we can compute in polynomial time, a list containing exactly one upper class machine from each composite machine and a 7-approximate integer solution which assigns small job configurations among the machines in the list. We use this result in section 3.4.4 to find the 7-approximate solution for the given Santa Claus instance within polynomial time. We close Section 3 with section 3.5  which gives the proof of our main result(Theorem \ref{thm1}) by combining the results established in the preceding subsections. We conclude the paper in Section 4 giving some areas where the methods specified in this work can be used to derive good results.      
 
\section{Preliminaries}
\label{sec:prelims}
Throughout this paper we have used terminology from the papers of Bansal and Svirideno \cite{bansal}, Asadpour \textit{et al.} \cite{AFS}\cite{AFS2}.  
In this paper, we consider the restricted Santa Claus problem and hence we use the term Santa Claus instance to refer to a restricted Santa Claus instance unless stated otherwise. 
We use two linear programs namely \emph{Assignment LP} and the \emph{Configuration LP}\cite{bansal} to solve the restricted Santa Claus problem. Each of the two linear programs give fractional solutions and these fractional solutions must be transformed to integral solutions without violating the constraints so as to yield a solution for the Santa Claus instance. 

\subsection{Assignment LP and Configuration LP}
\label{ALPCLP}
We denote The Assignment LP for any Santa Claus instance $I=\langle \mathcal{M},\mathcal{J},\mathcal{P}\rangle$ by the notation $ALP(I)$ and is shown in Figure \ref{AssignmentLP}. 
\begin{figure}[tbph]
\centering
\begin{alignat*}{3}
  & \text{maximize}   & \quad & \uptau          &&\nonumber \\
  & \text{subject to} &       & \sum\limits_{i \in \mathcal{M}}y_{ij}\le 1 &\hspace{10pt}& \forall j \in \mathcal{M}\\
  &                   &       & \sum \limits_{j\in \mathcal{J}}\mathcal{P}[i,j]y_{ij}\ge \uptau  & \hspace{10pt} & \forall i \in \mathcal{M} \\
  &                   &       & 0 \le y_{ij}\le 1  & \hspace{10pt} & \forall i \in \mathcal{M},\forall j \in \mathcal{J} \nonumber
\end{alignat*}
\caption{Assignment LP}
\label{AssignmentLP}
\end{figure}

Here the variable $y_{ij}$ indicates whether job $j$ is allocated to machine $i$ or not.A feasible solution of $ALP(I)$ is a vector $\mathbf{y}=\Big[y_{ij} \Big| i \in \mathcal{M}, j \in \mathcal{J}\Big]$ that satisfies the constraints. In particular, a feasible $\{0,1\}$ solution of $ALP(I)$ is a solution for the Santa Claus instance $I$. A fractional solution of $ALP(I)$ can be turned into a feasible $\{0,1\}$ solution by using the following rounding method due to Bezakova and Dani\cite{bezakova}.
\begin{lemma}
\label{bezakovadani}
Corresponding to a Santa Claus instance $I=\langle \mathcal{M},\mathcal{J},\mathcal{P}\rangle$, each fractional solution of $ALP(I)$ with objective function value $T$ can be converted to a feasible $\{0,1\}$ solution with objective function value $T-\max(\{\mathcal{P}[i,j]\mid i \in \mathcal{M}, j \in \mathcal{J}\})$ within polynomial time. 
\end{lemma}       
In order to describe the Configuration LP, we first introduce the concept of \emph{Configuration} or \emph{Bundle} which is nothing other than a set of jobs. Let $\uptau$ be the value of the objective function which we want to achieve. We call a configuration $C$ to be a minimal configuration of size $\uptau$ if and only if the removal of any job from $C$ makes the resulting configuration to have a size strictly less than $\uptau$.  Let $\mathcal{C}(i,\uptau)$ denote the set of all minimal configurations of size $\uptau$ for machine $i$.  

We denote the Configuration LP for the instance $I=\langle \mathcal{M},\mathcal{J},\mathcal{P} \rangle $ with minimal configurations of size $\uptau$ by the notation $CLP(I,\uptau)$ and is shown in Figure \ref{ConfigurationLP}.  
\begin{figure}[tbph]
\centering
 \begin{alignat*}{3}
  & \text{maximize}   & \quad & 0          &&\nonumber \\
  & \text{subject to} &       & \sum\limits_{C \in \mathcal{C}(i,\uptau)}\hspace{-5pt}x_{iC}\ge 1 &\hspace{10pt}& \forall i \in \mathcal{M}\\
  &                   &       & \sum \limits_{i \in \mathcal{M}}\bigg(\sum\limits_{\substack{C:C \in \mathcal{C}(i,\uptau)\\ j \in C}}\hspace{-10pt}x_{iC}\bigg)\le 1  & \hspace{10pt} & \forall j \in \mathcal{J} \\
  &                   &       &  0\le x_{iC}\le 1   & \hspace{10pt} & \forall i \in \mathcal{M}, \forall C \in \mathcal{C}(i,\uptau) \nonumber
\end{alignat*}
\caption{Configuration LP}
\label{ConfigurationLP}
\end{figure}

Here the variable $x_{iC}$ is denotes whether configuration $C$ is allocated to machine $i$ or not. The first set of constraints specifies that each machine must be allocated with at least one unit of configuration whereas the second set of constraints specifies that each job must be allocated to at most one machine. A feasible solution of this LP is a vector $\mathbf{x}=\Big[x_{iC} \Big | i \in \mathcal{M}, C \in C(i,\uptau)\Big]$ that satisfies the constraints. For any feasible solution $\mathbf{x}$ of $CLP(I,\uptau)$, we call the parameter $\sum\limits_{C:C \in \mathcal{C}(i,\uptau)}x_{iC}$ to be the \emph{total fractional weight of machine }$i$ with respect to the solution $\mathbf{x}$.

We see that the Configuration LP has exponentially many decision variables and polynomially many constraints.
Bansal and Sviridenko proposed a method to solve this LP in polynomial time\cite{bansal}. In this method, we treat the Configuration LP as the primal program and obtain the corresponding dual program which has exponentially many constraints, corresponding to the Configuration LP decision variables and polynomially many decision variables corresponding to the Configuration LP constraints. The separation algorithm for the dual is the minimum knapsack problem which can be solved in polynomial time using dynamic programming technique\cite{cormen,DP}. Hence the dual program can be solved in polynomial time using the ellipsoid method and a feasible solution of the corresponding primal program which is the Configuration LP, by considering the set of decision variables corresponding to the dual constraints considered by the ellipsoid method and then solving the Configuration LP on those set of decision variables in polynomial time.
 
We easily see that if $CLP(I,{\uptau}_0)$ is feasible, then $CLP(I,\uptau)$ is feasible for any $\uptau$ such that $\uptau \le {\uptau}_0$. The following Lemma shows that a fractional feasible solution can be turned to a fractional feasible solution of the Assignment LP.
\begin{lemma}
\label{configurationLPtoAssignmentLP}
Given Santa Claus instance $I=\langle \mathcal{M},\mathcal{J},P\rangle$. Each feasible solution of the $CLP(I,{\uptau}_0)$ induces a feasible solution for $ALP(I)$ with objective function at least ${\uptau}_0$. 
\end{lemma}
\begin{proof}
Let $\mathbf{x}$ be a feasible solution of $CLP(I,{\uptau}_0)$ for a given Santa Claus instance $I=\langle \mathcal{M},\mathcal{J},P\rangle$. Therefore $\mathbf{x}$ allocates the jobs in $\mathcal{J}$ to the machines in $\mathcal{M}$ in such a way that each machine receives job configurations with total size at least ${\uptau}_0$ . Now we form a vector $\mathbf{y}=[y_{ij}| i \in \mathcal{M}, j \in \mathcal{J}]$ where $y_{ij}$ is defined as $y_{ij}=\sum\limits_{C:C \in C(i,{\uptau}_0)}\Big(\sum\limits_{j \in C} x_{iC}\Big)$. Since $\mathbf{x}$ is a feasible solution of the Configuration LP, each machine $i \in \mathcal{M}$ is assigned at least one unit of configurations with total size of at least ${\uptau}_0$, we see that $\sum\limits_{j \in \mathcal{J}}y_{ij} \mathcal{P}[i,j] \ge {\uptau}_0$, thereby satisfying the first set of constraints of the Assignment LP with $\uptau={\uptau}_0$. Furthermore $\mathbf{x}$ assigns each job $j$ to at most one machine and hence $\sum\limits_{i \in \mathcal{M}} y_{ij}=\sum \limits_{i \in \mathcal{M}}\bigg(\sum\limits_{C:C \in C(i,{\uptau}_0)}\Big(\sum\limits_{j \in C} x_{iC}\Big)\bigg)\le 1$, thereby satisfying the second set of constraints of the Assignment LP. Hence we see that $\mathbf{y}$ is a feasible solution of $ALP(I)$ with objective function value ${\uptau}_0$.  Hence the Lemma.
\end{proof}
\subsection{$\alpha$-gap Instance, Big jobs and Small jobs}
\label{restrictedcase}
 Let $I=\langle \mathcal{M},\mathcal{J},\mathcal{P}\rangle$ be a Santa Claus instance and $\uptau$ be a positive real number such that $CLP(I,\uptau)$ is feasible. Let $\alpha$ be a positive real number. Now we define a new Santa Claus instance namely $\alpha$-\emph{gap instance} as follows.
\begin{definition}
\label{alphagapinstance}
The $alpha$-gap instance corresponding to the given instance $I$ is defined as the Santa Claus instance $I_{\alpha}=\langle \mathcal{M},\mathcal{J},\mathcal{P}_{\alpha}\rangle$ where $(i,j)^{th}$ entry of $\mathcal{P}_{\alpha}$ is given by 
\[\mathcal{P}_{\alpha}[i,j] =
\left\{
	\begin{array}{ll}
		\mathcal{P}[i,j]  & \mbox{if } \mathcal{P}[i,j] < \frac{\uptau}{\alpha} \\
		\uptau & \mbox{otherwise}
	\end{array}
\right. \]
\end{definition}
Hence we see that each job $j$ that has size at least $\frac{\uptau}{\alpha}$ on a subset of machines in instance $I$ has size $\uptau$ on the corresponding set of machines in the $\alpha$-gap instance $I_{\alpha}$. Now the following Lemma easily follows.
\begin{lemma}
\label{bigsmallgiftlemma}
A feasible solution of $CLP(I_{\alpha},\uptau)$ is a feasible solution for $CLP(I,\uptau)$ as well.
\end{lemma}
We call each job in $I$ as big job or small job with respect to the pair of parameters $\langle \uptau, \alpha \rangle $ as per the following definition.
\begin{definition}
\label{bigjobsmalljob}
A job $j$ in $I$ is said to be a big job with respect to the pair of parameters $\langle \uptau,\alpha \rangle$ if and only if $P_{\alpha}[i,j] \in \{0,\uptau\}$ for all machines $i \in \mathcal{M}$. A job is said to be a small job if and only if it is not a big job.
\end{definition}
We use the notation $\mathcal{J}_{big}$ and $\mathcal{J}_{small}$ to denote the set of big jobs and the set of small jobs in $\mathcal{J}$ respectively. We also use the notation $\mathcal{C}_{small}(i,\uptau)$ to denote the subset of $\mathcal{C}(i,\uptau)$ containing only small job configurations. Since each big job has size $\uptau$, we see that $|C_{big}(i,\uptau)| \le |\mathcal{J}_{big}|$.
\subsection{Forming clusters of machines and jobs}
\label{clusterformation}
We now convert any feasible solution of the configuration LP on an $\alpha$-gap instance $I_{\alpha}$ to another feasible solution for instance $I_{\alpha}$ itself with good structural properties as follows.

Let $\mathbf{x}$ be a feasible solution of $CLP(I_{\alpha},\uptau)$ and let $G=((L,R),E)$ be a weighted bipartite graph defined as 
\begin{dmath*}
{L =   \mathcal{M},R=\mathcal{J}_{big},
E = \Big\{ \{i,j\} \Big| x_{i\{j\}}>0,i \in \mathcal{M} \text{ and }  j \in \mathcal{J}_{big}\Big\} \text{ and }}\\
{\text{each edge } \{i,j\} \in E \text{ is given the weight }x_{i\{j\}}}
\end{dmath*}
It is easy to see that the graph $G$ can be constructed in polynomial time. The following Lemma is due to Bansal and Sviridenko\cite{bansal} that transforms the feasible solution $\mathbf{x}$ to another feasible solution by modifying the graph $G$.
\begin{lemma}
\label{cyclekillinglemma}
The feasible solution $\mathbf{x}$ of $CLP(I_{\alpha},\uptau)$ can be transformed to another feasible solution $\mathbf{x}^{*}$ of $CLP(I_{\alpha},\uptau)$ in polynomial time by modifying the graph $G$ to a forest $G^*$ using a cycle elimination procedure in such a way that ${x^*}_{iC} = x_{iC}$ for each $i \in \mathcal{M}$ and $C \in \mathcal{C}_{small}(i,\uptau)$.
\end{lemma} 
The following is another Lemma due to Bansal and Sviridenko\cite{bansal} that can be used to obtain a cluster of machines and jobs from the solution $\mathbf{x}^{*}$\footnote{This Lemma is a modified version of the original Lemma presented in Bansal and Sviridenko\cite{bansal}, but the modification follows from the proof presented in the paper.}
\begin{lemma}[\textsc{The Clustering Lemma}\cite{bansal}]
\label{clusteringlemma}
Using the forest $G^*$ and the solution $\mathbf{x}^{*}$, it is possible in polynomial time to cluster the machines and big jobs into groups ${M}_1,{M}_2,{M}_3,\cdots,{M}_r$ and ${J}_1,{J}_2,{J}_3,\cdots,{J}_r$ respectively with the following properties.
\begin{enumerate}
    \item For each $k=1,2,3\cdots,r$, $\mid {J}_k\mid=\mid {M}_k \mid -1$.
    \item For each $k=1,2,3\cdots,r$, the big jobs in ${J}_k$ may be placed feasibly on any of the $\mid{M}_i\mid-1$ machines in $M_k$.
    \item For each $k=1,2,3\cdots,r$, $\sum\limits_{i \in M_k}\bigg(\sum\limits_{C \in \mathcal{C}_{small}(i,T))}{x^*}_{iC}\bigg)\ge \frac{1}{2}$.
\end{enumerate}
\end{lemma}
This lemma provides completely flexible way of allocating big jobs to the machines in a cluster. Hence we can find a feasible solution of the Configuration LP in such a way that corresponding to each $k=1,2,3,\cdots,p$, the big jobs in $J_k$ are allocated to all the machines in $M_k$ except one machine which is satisfied by allocating one unit of small configurations. From now onwards, we call each set of machines ${M}_k$ as a \emph{Super Machine} and each set of jobs in $J_{k}$ as a \emph{Super Job}. 
\subsection{The Alternating Tree method for solving the restricted Santa Claus problem}
\label{alternatingtreesection}
A combinatorial superpolynomial time algorithm due to Asadpour \textit{et al.}\cite{AFS,AFS2} solves the restricted Santa Claus problem and obtains a $\delta-$ approximate solution for any $\delta\ge 4$. This method presents the problem as a  Hypergraph matching problem in a bipartite hypergraph $\mathcal{H}=((U,V),F)$ created using a feasible solution $\mathbf{x}$ of $CLP(I,\uptau)$ as per the following rules. The vertices in $U$ are labeled by the machines, the vertices in $V$ are labelled by the jobs in $V$. The set of hyperedges in $F$ are formed as per the following rules.
\begin{itemize}
\item For each machine $i$ and for each big configuration $C$ containing big job $j$ such that ${x}_{iC}>0$, we add a hyperedge containing exactly $i$ and $j$.
\item For each machine $i$ and for each small job configuration $C$ such that ${x}_{iC}>0$, then for each minimal configuration $C' \subset C$ of size $\frac{T}{\delta}$, we add a hyperedge containing exactly $i$ and $C'$.   
\end{itemize}  

We form the desired solution of the problem by finding a perfect matching that matches all the vertices in $U$. The method of formation of the matching is iterative and it uses an alternating tree which consists of two types of edges namely \emph{Add} edges which we would like to add to our current matching and \emph{Blocking} edges which are the edges in the current matching which block the insertion of add edges into the matching. The method starts with an empty matching and iteratively extends the current partial matching to an additional machine until all the machines get matched as per the following Lemma.
\begin{lemma}
\label{alternatingtreelemma}
In each step for extending the current matching, we can either match an additional machine by adding an edge which is not blocked by the edges in the current matching,  or augment the alternating tree with an edge  $e$ in $\mathcal{H}$ such that $e$ intersects with some machine vertex in the alternating tree and $e$ does not intersect with any job vertex in the alternating tree. 
\end{lemma}
Even though this method can find an integer solution whose integrality gap is close to the best known theoretical value, it takes superpolynomial local search moves to find the perfect matching\cite{AFS,AFS2}.  
  
\section{Clustered instances seem promising}
Our first step is to start with the super machines obtained by clustering method due to Bansal and Sviridenko \cite{bansal}.  We first give the Algorithm outline.
\subsection{Algorithm Outline}
\label{sec:algooutline}
\noindent
Let $I=\langle \mathcal{M},\mathcal{J},\mathcal{P}\rangle$ be a Santa Claus instance. Now our algorithm on input $I$ proceeds as per the following outline.
\begin{enumerate}
\item Using binary search over the range of possible values of $\uptau$, we solve $CLP(I,\uptau)$ to find the highest value $T$ such that $CLP(I,T)$ is feasible \cite{bansal}.
\item Form the 12-gap instance $I_{12}=\langle \mathcal{M},\mathcal{J},\mathcal{P}_{12}\rangle$ for the given instance $I$ as per Definition \ref{alphagapinstance} given in section \ref{restrictedcase}.
\item Classify the set of jobs $\mathcal{J}$ into two categories -$\mathcal{J}_{big}$ and $\mathcal{J}_{small}$ containing the big jobs and the small jobs with respect to the parameters $\langle T, 12 \rangle $ as per Definition \ref{bigjobsmalljob} given in section \ref{restrictedcase}.
\item Solve $CLP(I_{12},T)$ resulting in feasible solution $\mathbf{x}$. We then classify the machines in $\mathcal{M}$ into two categories, upper class machines and middle class machines. Each upper class machine $i$ places at least $\frac{1}{2}$ of its total fractional weight from the solution $\mathbf{x}$ on big job configurations where as each middle class machine $i$ places at least $\frac{1}{2}$ of its total fractional weight from $\mathbf{x}$  on small job configurations. This step is detailed in section \ref{step1}.
\item If $I$ has no upper class machines, then each machine in the given instance $I$ places at least half of its total fractional weight on small jobs. Hence we find the 12-approximate solution which allocates only small job configurations to all the machines by solving the corresponding  Assignment LP
and then rounding the resulting fractional solution using Lemma \ref{bezakovadani} due to Bezakova and Dani\cite{bezakova}.
We detail this step in section \ref{case1}.
\item If $I$ has one or more upper class machines, then we do the following. 
\begin{enumerate}
\item Form the set of super machines $\mathcal{SM}$ out of the upper class machines and the set of composite machines $\mathcal{CM}$ using the machine job clustering technique due to Bansal and Sviridenko\cite{bansal} satisfying the conditions in Lemma \ref{clusteringlemma} given in section \ref{clusterformation}. The details of this step is given in section \ref{formingsupermachines}. 
\item Form the Extended Configuration LP(\emph{ECLP}) given in section \ref{existenceofTby6} which is an extension of the Configuration LP over composite machines. We observe that this configuration LP can be solved in polynomial time using the now well-known technique in \cite{bansal}. We use this solution to find a pair $\langle L,\mathbf{u} \rangle $ where $L$ is a list containing all the middle class machines and exactly one upper class machine from each super machine in $\mathcal{SM}$ and $\mathbf{u}$ is an fractional feasible solution that allocate small job configurations of size at least $\frac{T}{12}$ among the machines in $L$. 
We detail this step in section \ref{computingTby6solution}. The existence of such a solution is guaranteed by Theorem \ref{thm2} by proving the existence of an integer solution using the alternating tree technique due to Asadpour \textit{et al.}\cite{AFS}\cite{AFS2} given in section \ref{alternatingtreesection}. 
The details of Theorem \ref{thm2} is given in section \ref{existenceofTby6}. 
\item A particularly interesting result that we present is the conversion of a feasible solution to the ECLP to a feasible solution $\langle u \rangle$ of a bilinear  program which we call the Extended Assignment Problem (EAP).  The solution to the EAP plays a crucial role in obtaining an integer solution of size $\frac{T}{12}$ for the machines in the list $L$.  This ensures that we can assign small jobs among the machines in the list $L$ in such a way that each machine gets a total size of at least $\frac{T}{12}$.
This is presented in section \ref{computingTby6solution}.
\item Finally,
the upper class machines in each super machine which are not in list $L$ are assigned single big job from the corresponding super job. The resulting solution ensures that each machine in the given instance $I$ is assigned job configurations of size at least $\frac{T}{12}$. This step is detailed in section \ref{findingTby12solutionforcase2}.
\end{enumerate}
\end{enumerate}
In the following sections we prove properties of the correctness of all the steps in the outline presented above, and put them all together in the main theorem.

\subsection{Classifying the machines in the given instance}
\label{step1}
With reference to the solution $\mathbf{x}$, we now classify the set of machines $\mathcal{M}$ into two categories- $\mathcal{M}_{upper}$ and $\mathcal{M}_{middle}$ containing the upper class machines and middle class machines respectively as per the following definition. 
\begin{definition}
\label{upperclassmiddleclass}
A machine $i$ in $\mathcal{M}$ is said to be an {Upper Class Machine} if $\sum\limits_{\substack{C:C \in \mathcal{C}_{big}(i,T)}}x_{iC}\ge \frac{1}{2}$. A machine $i \in \mathcal{M}$ is said to be a {Middle Class machine} if it is not an upper class machine. That is a machine $i$ in $\mathcal{M}$ is said to be {upper class machine} or {middle class machine} according to whether the cumulative weight of the big gift configurations allocated to $i$ wth respect to the solution $\mathbf{x}$ is at least $\frac{1}{2}$ or not.
\end{definition}
Therefore,  for every middle class machine $i \in \mathcal{M}$, $\sum\limits_{\substack{C:C \in \mathcal{C}_{small}(i,T)}}x_{iC}\ge \frac{1}{2}$.  
\subsection{Dealing with instances having no Upper Class machines}
\label{case1}
We now show that if the given Santa Claus instance has no upper class machines with respect to $T$, then each machine in $I$ can be assigned small job configurations of size at least $\frac{T}{12}$. 
\begin{lemma}
\label{noupperclasslemma}
If there is no upper class machine, then a solution for the given Santa Claus instance $I$ that assigns small job configurations of size at least $\frac{T}{12}$ can be found in polynomial time. 
\end{lemma}    
\begin{proof}
Let the given Santa Claus $I$ be such that there is no upper class machine in instance $I_{12}$ with reference to the  $CLP(I_{12},T)$ solution $\mathbf{x}$. This implies that $\sum\limits_{\substack{C:C\in \mathcal{C}_{small}(i,T)}}x_{iC}\ge \frac{1}{2}$ for each machine $i$ in $I_{12}$. Hence by Lemma \ref{configurationLPtoAssignmentLP}, there exists fractional feasible solution $\mathbf{y}$ for $ALP(I_{12})$ that fractionally allocates small jobs to each machine in $I$ such that each machine gets size at least $\frac{T}{2}$. Now applying Lemma \ref{bezakovadani}, we can convert the fractional solution $\mathbf{y}$ to a solution $\mathbf{y}'$ for $I_{12}$ that allocates small jobs to each machine in $I_{12}$ such that every machine gets size at least $\frac{T}{2}-\frac{T}{12}=\frac{5T
}{12}\ge \frac{T}{12}$. Since the small jobs in both $I$ and $I_{12}$ are the same and have the same size, we see that $\mathbf{y}'$ is a solution for instance $I$ as well. Hence we claim that there exists a solution for the given instance $I$ such that each machine in $I$ is assigned small job configurations of size at least $\frac{T}{12}$.

Clearly the solution $\mathbf{y}'$ can be obtained by solving the Assignment LP on the subinstance $I'$ of $I_{12}$ that consists of all the machines and only the smal jobs in $I$. Since solving $ALP(I')$ and finding the integer solution mentioned in Lemma \ref{bezakovadani} can be found in polynomial time, the solution $\mathbf{y}'$ can be found in polynomial time. Hence the proof.     
\end{proof} 

\subsection{Dealing with instances having one or more upper class machines}  
We now consider the case where there are one or more upper class machines in $\mathcal{M}$ with respect to $\mathbf{x}$. 
\subsubsection{Forming the Super Machines and Composite Machines}
\label{formingsupermachines}
Let $I'_{12}$ be the instance obtained by $I_{12}$ such that $I'_{12}$ consists of the upper class machines in $\mathcal{M}$ and the big jobs which are fractionally allocated to the upper class machines by the fractional solution $\mathbf{x}$. Let $\mathbf{x}'$ be the restriction of $\mathbf{x}$ with respect to the instance $I'_{12}$. Now we form the graph $G$ mentioned in section \ref{clusterformation} and apply the Lemma \ref{cyclekillinglemma} on $G$ to get the fractional feasible solution $\mathbf{x^*}$ mentioned in the Lemma. We then apply Lemma \ref{clusteringlemma} on $\mathbf{x}^{*}$, to form the set of super machines denoted by $\mathcal{SM}$. In the remainder of this paper, we use the notation $\mathcal{CM}$ to denote the set $\mathcal{SM} \cup \Big\{\{i\} \Big| i \in \mathcal{M}_{middle} \Big\}$ and call each member of the set $\mathcal{CM}$ to be a \emph{Composite Machine}. We note that each composite machine is either a super machine or a singleton set consisting of a middle class machine. With this convention we unambiguously reason about machines inside a composite machine.

\subsubsection{Existence of a $\frac{T}{6}$ size integer solution for the Composite machines}
\label{existenceofTby6} 
We now show that there exists an integer solution that allocates small job configurations among the composite machines in such a way that at least one machine from each composite machine in $\mathcal{CM}$ gets a total size of at least $\frac{T}{6}$. 

Now we define the \emph{Modified Configuration LP} which is a modification of the Configuration LP in two respects. 
\begin{enumerate}
\item It is specifically written for composite machines rather than single machines.
\item All the jobs considered in it are small jobs. 
\item The first set of constraints in it requires that every composite machine should have a total weight of at least $\frac{1}{2}$ from small job configurations. This means that each of the super machines in $\mathcal{SM}$ and each middle class machine in $\mathcal{M}_{middle}$ must have a total fractional weight of at least $\frac{1}{2}$ from small jobs.
\end{enumerate}
\begin{figure}[tbph]
\centering
\begin{alignat*}{3}
  & \text{maximize}   & \quad & 0          &&\nonumber \\
  & \text{subject to} &       & \sum\limits_{i \in D}\bigg(\sum\limits_{C \in \mathcal{C}_{small}(i,T)}z_{iC}\bigg)\ge \frac{1}{2} &\hspace{10pt}& \forall D \in \mathcal{CM}\\
  &                   &       & \sum\limits_{i \in \mathcal{M}}\bigg(\sum\limits_{\substack{C: \\ C \in \mathcal{C}_{small}(i,T)\\ j \in C}}z_{iC}\bigg)\le 1  & \hspace{10pt} & \forall j \in \mathcal{J}_{small} \\
  &                   &       &  0\le z_{iC}\le 1   & \hspace{10pt} & \forall i\in \mathcal{M}\text{,} \forall C \in \mathcal{C}_{small}(i,T) \nonumber
\end{alignat*}
\caption{Modified Configuration LP}
\label{ModifiedConfigurationLP}
\end{figure}

Now we prove the following Lemma.
\begin{lemma}
\label{modifiedconfigLP}
The Modified Configuration LP is feasible.
\end{lemma}
\begin{proof}
Consider the fractional solution $\mathbf{x}^{*}$ computed above as a result of applying Lemma \ref{cyclekillinglemma}. From the construction of the feasible solution $\mathbf{x}^{*}$ of $CLP(I_{12},T)$ computed above and Lemma \ref{cyclekillinglemma}, it is clear by definition that corresponding to each middle class machine $i$ and $C \in \mathcal{C}_{small}(i,T)$, $\sum\limits_{C \in \mathcal{C}_{small}}{x^*}_{iC}\ge \frac{1}{2}$. This implies that corresponding to each composite machine $D$ representing a middle class machine $i$, we must have $\sum\limits_{C \in \mathcal{C}_{small}}{x^*}_{iC}\ge \frac{1}{2}$. Furthermore it follows from Lemma \ref{clusteringlemma} that corresponding to each composite machine representing some super machine $M$, $\sum\limits_{i \in M}\bigg(\sum\limits_{C \in \mathcal{C}_{small})}{x^*}_{iC}\bigg)\ge \frac{1}{2}$. Hence we see that $\sum\limits_{i \in D}\bigg(\sum\limits_{C \in \mathcal{C}_{small})}{x^*}_{iC}\bigg)\ge \frac{1}{2}$ for each composite machine $D \in \mathcal{CM}$, thereby satisfying the first set of constraints of the Modified Configuration LP. It also follows from the construction of the solution $\mathbf{x}^{*}$ and Lemma \ref{cyclekillinglemma} that the fractional allocation of a small job configuration $C$ to a machine $i$ in $\mathbf{x}^{*}$ is not more than the same in the initial feasible solution $\mathbf{x}$ which we used to classify the set of machines into upper class machines and middle class machines. Hence it becomes clear that \[\sum\limits_{i \in \mathcal{M}}\bigg(\sum\limits_{\substack{C: \\ C \in \mathcal{C}_{small}(i,T)\\ j \in C}}{x^*}_{iC}\bigg)\le \sum\limits_{i \in \mathcal{M}}\bigg(\sum\limits_{\substack{C: \\ C \in \mathcal{C}_{small}(i,T)\\ j \in C}}x_{iC}\bigg)\le 1 \text{ for each small job } j \in \mathcal{J}_{small}.\] Hence $\mathbf{x}^{*}$ satisfies the second set of constraints of the Modified Configuration LP. Hence we see that the solution $\mathbf{x}^{*}$ satisfies all the constraints of the Modified Configuration LP and is therefore a feasible solution of this LP. Hence the proof.     
\end{proof} 
  
Now in order to prove the existence of a $\{0,1\}$ solution that allocate small job configurations of size at least $\frac{T}{6}$ among the composite machines, we use the alternating tree method\cite{AFS}, where we pose the problem as a hypergraph matching problem in a suitably defined bipartite hypergraph $\mathcal{H}$ using the solution $\mathbf{x}^{*}$ mentioned in Lemma \ref{modifiedconfigLP} and then find a perfect matching in $\mathcal{H}$ by forming the alternating tree and growing it as stated in Lemma \ref{alternatingtreelemma}.    

We define the bipartite hypergraph $\mathcal{H}=((U,V),F)$ with some differences. Here the vertices in $U$ are labelled by the composite machines in $\mathcal{CM}$, the vertices in $V$ are labelled by small jobs and the hyperedges in $F$ are formed as per the following rules.
\begin{itemize}
\item For each composite machine $D$ representing a middle class machine $i$ and for each small configuration $C$, if ${x}^{*}_{iC}>0$, then for each minimal configuration $C'$ of size $\frac{T}{6}$ such that $C' \subset C$, we add a hyperedge connecting $D$ and $C'$. 
\item For each composite machine representing a super machine $M$ and for each small configuration $C$, if there exists an upper class machine $i$ within the super machine $M$ such that ${x^*}_{iC}>0$, then for each minimal configuration $C'$ of size $\frac{T}{6}$ such that $C' \subset C$, we add a hyperedge containing $i$ and $C'$.   
\end{itemize}  
By the way of formation of the hyperedges within $\mathcal{H}$, we see that the size of the jobs present in each hyperedge is strictly less than $\frac{T}{6}+\frac{T}{12}$. We now show that the bipartite hypergraph $\mathcal{H}$ constructed above has a perfect matching.  Further, note that unlike in Asadpour \textit{et.al} \cite{AFS}\cite{AFS2}, we do not have any big job in the hypergraph.  They have all been taken care of in Bansal and Sviridenko's clustering step.  
\begin{lemma}
\label{perfectmatchinglemma}
The bipartite hypergraph $\mathcal{H}$ constructed above has a perfect matching that matches each composite machine with a small job configuration having size at least $\frac{T}{6}$. 
\end{lemma}
\begin{proof}
We prove the Lemma using the method of contradiction. In particular, we prove that a given partial matching that does not match all the composite machines in $\mathcal{CM}$ can be extended to match an additional machine which in turn implies that $\mathcal{H}$ has a perfect matching. The way we do this is to use an argument identical to the one in Asadpour \textit{et al.}\cite{AFS,AFS2} which we have presented in Lemma \ref{alternatingtreelemma}. They show that if the alternating tree can be augmented with a hyperedge, that is a new hyperedge can be added to the tree which out intersecting any hyperedge already in the tree, then the constructive procedure will eventually match one additional machine. Therefore, we prove that the alternating tree can be augmented with one hyperedge from $\mathcal{H}$.  In the remainder of this proof, for a set $W$ of edges in the graph $\mathcal{H}$ we use the notation $W_U$ and $W_V$ to denote the set of vertices of $U$ which are incident at some edge in $W$ and the set of vertices of $V$ which are incident at some edge in $W$, respectively.

Assume that, at some step in growing the alternating tree, a new hyperedge from $\mathcal{H}$ cannot be added to alternating tree without intersecting a hyperedge that is already in the alternating tree.  
Let $A,B$ and $\mathcal{CM}'$ respectively denote the set of Add edges in the alternating tree, the set of Blocking edges in the alternating tree, and the set of vertices of $U$ matched by the current matching at this step. Let $X$ denote the total size of the jobs present in the edges within the edges in the alternating tree and $\gamma=\sum\limits_{D \in \mathcal{CM}'}\bigg(\sum\limits_{i \in D}\Big(\sum\limits_{C: C \in {C}_{small}(i,T)}{x^*}_{iC}\times \text{ size of C} \Big)\bigg)$.  together with the fact that the alternating tree cannot be grown further, we see that it is not possible to find an edge $e$ in $\mathcal{H}$ such that $e$ is incident at some machine vertex in $(A \cup B)_{U}$ and does not intersect with any of the job vertices in $(A \cup B)_{V}$. This further implies that the quantity $X$ representing the total size of the jobs appearing in the edges of the alternating tree satisfies the condition
\begin{eqnarray}
\label{relation1}
X & = & \sum\limits_{D \in \mathcal{CM}'}\Bigg(\sum\limits_{i \in D}\bigg(\sum\limits_{C:C \in \mathcal{C}_{small}(i,T)}{x^*}_{iC}\Big(\sum\limits_{j: j \in C \cap (A \cup B)_{V}} \mathcal{P}_{\alpha}[i,j] \Big) \bigg)\Bigg) >\frac{5}{6}\gamma
\end{eqnarray}   
Another way of accounting for the $X$ is obtained by taking the sum of the size of the jobs within the edges in $A$, and in each blocking hyperedge $e$ in $B$, we take the size of the jobs which are not present in the Add  edge in $A$ that is blocked by $e$. Due to the minimality of the hyperedges in $\mathcal{H}$, we also see that each hyperedge within $A$ contains as many jobs as to give a maximum size of $\frac{T}{6}+\frac{T}{12}$ and for each blocking hyperedge $e$ in $B$, the total size of the jobs which are present in $e$ and not present in the Add edge in $A$ that is blocked $e$ is at most $\frac{T}{6}$. In other words, each hyperedge in $A$ contributes at most $\frac{T}{6}+\frac{T}{12}$ and each hyperedge in $B$ contributes at most $\frac{T}{6}$ to $X$. Thus we have 
\begin{eqnarray*}
X & \le & |A|\Big(\frac{T}{6}+\frac{T}{12}\Big)+\frac{T}{6}|B| 
\end{eqnarray*}  
Since every Add edge in $A$ is blocked by at least one hyperedge edge in $B$ and all these blocking hyperedges are different, we must have $|A|\le |B|$. This implies that $X  \le  \frac{5T}{12}|B|$. Using this relation in (\ref{relation1}), we get $\frac{5T}{12}|B| > \frac{5}{6} \gamma$. Hence we have
\begin{eqnarray*}\\
\frac{T}{2}|B| &> & \sum\limits_{D \in \mathcal{CM}'}\bigg(\sum\limits_{i \in D}\Big(\sum\limits_{C: C \in C_{small}(i,T)}{x^*}_{iC}\times \text{ size of C} \Big)\bigg)
\end{eqnarray*}
Since the size of each configuration in $C_{small}$ is at least $T$, we see that
\begin{eqnarray*}
\frac{T}{2}|B| &> & T\sum\limits_{D \in \mathcal{CM}'}\bigg(\sum\limits_{i \in D}\Big(\sum\limits_{C: C \in C_{small}(i,T)}{x^*}_{iC} \Big)\bigg)
\end{eqnarray*}
which in turn implies that 
\begin{eqnarray}
\label{relation2}
\frac{1}{2}|B| &> & \sum\limits_{D \in \mathcal{CM}'}\bigg(\sum\limits_{i \in D}\Big(\sum\limits_{C: C \in C_{small}(i,T)}{x^*}_{iC}\Big)\bigg)
\end{eqnarray}
Since the number of composite machines in $\mathcal{CM'}$ is at least $|B|+1$, it is evident from relation (\ref{relation2}) that there exists a composite machine $D' \in \mathcal{CM}'$ such that $\sum\limits_{i \in D'}\Big(\sum\limits_{C: C \in C_{small}(i,T)}{x^*}_{iC}\Big)<\frac{1}{2}$. But this contradicts the fact that $\mathbf{x}^{*}$ is a feasible solution of the Modified Configuration LP and hence it satisfies the constraint $\sum\limits_{i \in D'}\Big(\sum\limits_{C: C \in C_{small}(i,T)}{x^*}_{iC}\Big)\ge \frac{1}{2}$. Hence we have proved that  the alternating tree can be augmented at any step if it cannot increase the size of the matching.
Thus, using Lemma \ref{alternatingtreelemma} it follows that we can find a perfect matching in $H$. Hence the proof.
\end{proof}
\noindent
Using Lemma \ref{perfectmatchinglemma}, we now prove the existence of the integer solution mentioned in the beginning of this section.  Next we define a new linear program namely \emph{Extended Configuration LP} which is given in Figure \ref{ECLP}.
\begin{figure}[tbph]
\centering
\begin{alignat*}{3}
  & \text{maximize}   & \quad & 0          &&\nonumber \\
  & \text{subject to} &       & \sum\limits_{i \in D}\bigg(\sum\limits_{C \in \mathcal{C}_{small}(i,\frac{T}{6})}z_{iC}\bigg)=1 &\hspace{10pt}& \forall D \in \mathcal{CM}\\
  &                   &       & \sum\limits_{i \in \mathcal{M}}\bigg(\sum\limits_{\substack{C: \\ C \in \mathcal{C}_{small}(i,\frac{T}{6})\\ j \in C}}z_{iC}\bigg)\le 1  & \hspace{10pt} & \forall j \in \mathcal{J}_{small} \\
  &                   &       &  0\le z_{iC}\le 1   & \hspace{10pt} & \forall i\in \mathcal{M}\text{,} \forall C \in \mathcal{C}_{small}(i,\frac{T}{6}) \nonumber
\end{alignat*}
\caption{Extended Configuration LP}
\label{ECLP}
\end{figure}
\begin{theorem}
\label{thm2}
There exists a $\{0,1\}$ feasible solution for ECLP.
\end{theorem}
\begin{proof}
It follows from Lemma \ref{perfectmatchinglemma} that there exists a perfect matching in the graph $\mathcal{H}$ constructed as specified above such that the matching matches every composite machine with a small configuration of size at least $\frac{T}{6}$. Let $\mathscr{M}$ be this perfect matching. For each composite machine $D \in \mathcal{CM}$, let $C_D$ denote the small job configuration such that the matching $\mathscr{M}$ matches the composite machine $D$ with configuration $C_D$.  Now we define the vector $\mathbf{z}'$ as follows.
\begin{itemize}
\item For each middle class machine $i$ represented by the composite machine $D \in \mathcal{CM}$, we set $z'_{iC_D}=1$ and $z'_{iC}=0$ for each $C \in \mathcal{C}_{small}(i,\frac{T}{6})$ such that $C \neq C_D$.
\item For each super machine $M$ represented by the composite machine $D \in \mathcal{CM}$, let $i' \in M$ be the upper class machine that is matched in the perfect matching $\mathscr{M}$.
Then we set $z'_{i'C_D}=1$ and $z'_{i'C}=0$ for each $C \in \mathcal{C}_{small}(i,\frac{T}{6})$ such that $C \neq C_D$. Furthermore, for each machine $i \in M$ such that $i \neq i'$, we set $z'_{iC}=0$ for all $C \in \mathcal{C}_{small}(i,\frac{T}{6})$.
\end{itemize} 
\begin{claim}
\label{claim1}
$\mathbf{z}'$ is a $\{0,1\}$ feasible solution of ECLP.
\end{claim}
\begin{proof}
From the construction of $\mathbf{z}'$, we note the following.
\begin{itemize}
\item Corresponding to each composite machine $D$ there is exactly one machine $i'$ such that $z'_{i'C_D}=1$ and $z'_{iC}=0$ for each $C \in \mathcal{C}_{small}(i,\frac{T}{6})$ such that $C \neq C_D$.  Furthermore, for each $i \in M$ such that $i \neq i'$, we have $z'_{iC}=0$ for each $C \in \mathcal{C}_{small}(i,\frac{T}{6})$. This implies that $\sum\limits_{i \in D}\bigg(\sum\limits_{C \in \mathcal{C}_{small}(i,\frac{T}{6})}z_{iC}\bigg)=1$. Hence $\mathbf{z}'$ satisfies the first set of constraints of \emph{ECLP}.
\item No small job is allocated to more than one machine and hence $\mathbf{z}'$ satisfies the second set of constraints of \emph{ECLP}.
\item $z'_{iC} \in \{0,1\}$ for each $ i\in \mathcal{M}$ and $C \in \mathcal{C}_{small}(i,\frac{T}{6})$.
\end{itemize}
Hence the claim.
\end{proof}
Now by Claim \ref{claim1}, the Lemma immediately follows.
\end{proof}

\noindent
We see that similar to Configuration LP, \emph{ECLP} has exponentially many decision variables and polynomially many constraints. Therefore we can use the technique due to Bansal and Sviridenko \cite{bansal} which is mentioned in section \ref{ALPCLP} for solving the Configuration LP to solve \emph{ECLP} as well in polynomial time. We easily see that the seperation oracle for the dual program corresponding to \emph{ECLP} is also the minimum knapsack problem. Hence \emph{ECLP} can also be solved in polynomial time exactly in the same way as Configuration LP is solved in polynomial time. 

\subsubsection{Efficient Computation of a $\frac{T}{6}$ size fractional solution for the composite machines}
\label{computingTby6solution}
By Theorem \ref{thm2}, it is clear that there exists an integer  solution that assigns small job configurations to the composite machines in such a way that at least one of the machines in each composite machine gets a total size of at least $\frac{T}{6}$. 
We propose a method for using the existence of such an integer solution in a polynomial time algorithm. Towards this end, we define a new mathematical program namely \emph{Extended Assignment Program}(\emph{EAP}) and is given in Figure \ref{fig:EAP}. 
\begin{figure}[!ht]
\centering
\begin{alignat}{3}
  & \text{minimize}   & \quad & 0          &&\nonumber \\
  & \text{subject to} &       & \sum\limits_{i \in D}\Big(s_i\sum\limits_ {j \in \mathcal{J}_{small}}u_{ij}\mathcal{P}_{12}[i,j]\Big)\ge \frac{T}{6} &\hspace{10pt}& \forall D \in \mathcal{CM} \label{eqn:EAP1}\\
  &                   &       & \sum\limits_{i \in D}s_i= 1  & \hspace{10pt} & \forall D \in \mathcal{CM} \label{eqn:EAP2} \\
   &                   &       & \sum\limits_{i \in \mathcal{M}}u_{ij}\le 1  & \hspace{10pt} & \forall j \in \mathcal{J}_{small} \label{eqn:EAP3}\\
  &                   &       &  0\le u_{ij}\le 1   & \hspace{10pt} & \forall i\in \mathcal{M}\text{,} \forall j \in \mathcal{J}_{small} \label{eqn:EAP4} \\
  &                   &       &  0\le s_i\le 1   & \hspace{10pt} & \forall i\in \mathcal{M}  
\end{alignat}
\caption{Extended Assignment Program}
\label{fig:EAP}
\end{figure} 
This mathematical program consists of a decision variable $u_{ij} \in [0,1]$ which indicates whether small job $j$ is allocated to machine $i$ or not. In addition, \emph{EAP} consists of a decision variable $s_i \in [0,1]$ for each machine $i \in \mathcal{M}$ which further controls the allocation of small jobs within machine $i$. The first set of constraints of the program states that the sum of the total size received by the machines in each composite machine in $\mathcal{CM}$ from small gifts must be at least $\frac{T}{6}$. Similarly the second set of constraints states that no small job is allocated to more than one machine. 
We first show that  the integer solution guaranteed by Theorem \ref{thm2} gives an integer assignment to {\em EAP} in  Figure \ref{fig:EAP}.
\begin{lemma}
 \label{lem:intsolution}
 {\em EAP} in Figure \ref{fig:EAP} and the vector program in Figure \ref{fig:SDPEAP} both have an integer solution.
\end{lemma}
\begin{proof}
 Let $\mathbf{z}$ be an integer solution to ECLP as guaranteed by Theorem \ref{thm2}.
 We define the integer values to the variables in $\langle \mathbf{u}, \mathbf{s} \rangle$ corresoponding to each composite machine as follows:
 For a composite machine $D$, let $i \in D$ be such that $z_{iC} = 1$ for a $C \in \mathcal{C}_{small}(i,\frac{T}{6})$.
 Define $s_{i} = 1$ and $s_{i'} = 0$ for all $i' \in D$ such that $i' \neq i$.
 For each $j \in C$ define $u_{ij} = 1$ and for all other $i' \in D$ and $j \in \mathcal{J}_{small}$, define $u_{i'j}=0$. 
 Since the value of the jobs in configuration $C$ is at least $\frac{T}{6}$, it follows that $\langle \mathbf{u} , \mathbf{s} \rangle$ is feasible for {\em EAP}  in Figure \ref{fig:EAP}. 
 
 We define the feasible solution for the vector program in Figure \ref{fig:SDPEAP} using the above solution defined for {\em EAP} as follows: all the vectors are $n$-dimensional vectors. $p$ is the vector in which each coordinate is 1, and $q_1, \ldots, q_n$ is the standard basis $e_1, \ldots, e_n\}$ of the $n$-dimensional Euclidean space.  If $s_i=1$, $a_i$ is $e_1$, a vector in the standard basis.   If $u_{ij}=1$, then $b_{ij} = e_1$.  All other vectors are zero vectors.  Since this assignment is defined from an integer solution for {\em EAP}, it is  a feasible solution 
 for the vector program in Figure \ref{fig:SDPEAP}.  Hence the lemma.
\end{proof}
\begin{lemma}
\label{lem:lemma0}
A fractional solution for \emph{EAP} can be found in polynomial time. 
\end{lemma}
\begin{proof}
We  show that a solution to the \emph{EAP} can be obtained from the  vector program in Figure \ref{fig:SDPEAP}. 
\begin{figure}[ht]
\centering
\begin{alignat}{3}
  & \text{minimize}   & \quad & 0          &&\nonumber \\
  & \text{subject to} &       & \sum\limits_{i \in D}\Big(\sum\limits_ {j \in \mathcal{J}_{small}}(a_i^T \cdot b_{ij}) \mathcal{P}_{12}[i,j]\Big)\ge \frac{T}{6} &\hspace{10pt}& \forall D \in \mathcal{CM}\\
  &                   &       & \sum\limits_{i \in D} {p}^T \cdot a_i = 1  & \hspace{10pt} & \forall D \in \mathcal{CM} \\
   &                   &       & \sum\limits_{i \in \mathcal{M}} {p}^T \cdot b_{ij} = 1  & \hspace{10pt} & \forall j \in \mathcal{J}_{small}  \\
  &                    &       & a_i^T \cdot q_l \geq 0,~b_{ij}^T \cdot q_l  \geq 0 &  &\forall ~l \in [n], \forall i \in \mathcal{M}, \forall j \in \mathcal{J}_{small} \\
  &                    &       & p^T \cdot p = n, ~p^T \cdot q_l = 1, ~q_l^T \cdot q_l = 1, ~q_l^T \cdot q_{l'} = 0 & \hspace{10pt}& \forall l \neq l' \in [n] \label{eqn:cEAP4}
\end{alignat}
\caption{A Vector Program for the Extended Assignment Program}
\label{fig:SDPEAP}
\end{figure} 
In the vector program each variable is a $n$-dimensional column vector. Recall, $n$ is the number of jobs.  From Equation \ref{eqn:cEAP4} it follows that in a solution to the vector program, the vectors $q_1 \ldots q_n$ forms an orthonormal basis of unit vectors of the $n$-dimensional Euclidean space.  In other words $q_1 \ldots q_n$
can be considered to define the axes of a coordinate system.  
Using standard techniques from linear algebra, in polynomial we can transform this coordinate system to a coordinate system in which the axes are defined by the standard basis vectors $e_1, \ldots, e_n$.  
 Clearly, all  pairwise relationships between two vectors in the original solution will be preserved in the transformed solution. Thus, $p$ will be transformed to the vector in which all coordinates are 1 (denoted by $\mathbf{1}$). Further, all the inequalities will be satisfied by the transformed solution.  Consequently, in the analysis below we consider a feasible solution in which $p$ is the vector $\mathbf{1}$ and $q_1, \ldots, q_n$ are the vectors $e_1, \ldots, e_n$ of the standard basis, respectively.  \\
From Lemma \ref{lem:intsolution} we know that the vector program as a feasible solution.  Let $\langle \mathbf{a}, \mathbf{b} \rangle$ denote a feasible solution for the vector program.
We show that from $\langle \mathbf{a}, \mathbf{b} \rangle$ we can get a feasible solution $\langle \mathbf{u}, \mathbf{s} \rangle$ for {\em EAP}.  The feasible solution is defined as follows for each composite machine $D \in \mathcal{CM}, i \in D, j \in \mathcal{J}_{small}$:
$$s_i = \mathbf{1}^T \cdot a_i, ~~ u_{ij} = \mathbf{1}^T \cdot b_{ij}$$
Let  $\theta_{ij}$ \text{ denote the angle between } $a_i$ \text{ and } $b_{ij}$.
To see that $\langle \mathbf{u}, \mathbf{s} \rangle$ is feasible for {\em EAP}, we first observe that $0 \leq s_i, u_{ij} \leq 1$. Secondly, we know that for a vector $v$, $\norm{v}_2 \leq \norm{v}_1$. In other words, the Euclidean norm of $v$ is at most the rectilinear norm of $v$.  Further, all the coordinates of the vectors $a_i$ and $b_{ij}$ are non-negative.  Thus, $\norm{a_i}_2 \leq \norm{a_i}_1 = \mathbf{1}^T \cdot a_i$ and $\norm{b_{ij}}_2 \leq \norm{b_{ij}}_1 = \mathbf{1}^T \cdot b_{ij}$. Therefore, it follows that 
$a_i^T \cdot b_{ij} = \norm{a_i}\cdot\norm{b_{ij}}\cdot cos(\theta_{ij}) \leq (\mathbf{1}^T \cdot a_i) \cdot (\mathbf{1}^T \cdot b_{ij}) = s_i \cdot u_{ij}$.  This ensures that Equation \ref{eqn:EAP1} is satisfied.  
Since $s_i$ is defined to be $\mathbf{1}^T \cdot a_i$, clearly Equation \ref{eqn:EAP2} is satisfied.
Further, $u_{ij} = \mathbf{1}^T \cdot b_{ij}$, and this ensures that Equation \ref{eqn:EAP3} is satisfied.\\ Finally,  it is well-known that the vector program in Figure \ref{fig:SDPEAP} matches the definition of a standard semi-definite program and a solution to the vector program can be found  standard semi-definite programming techniques.  Hence the lemma.
\end{proof}
\noindent
Next we show the following lemma which shows that from the feasible solution $\langle \mathbf{u},\mathbf{s}\rangle$ to the \emph{EAP} we can {\em extract} an Assignment LP instance which is feasible for the value of $\frac{T}{6}$.
\begin{lemma}
\label{roundinglemma}
Corresponding to the feasible solution $\langle \mathbf{u},\mathbf{s}\rangle$ of the EAP described above, there is a polynomial time computable feasible solution $\langle \mathbf{u}',\mathbf{s}'\rangle$ for EAP such that $\mathbf{s}' \in \{0,1\}^{|\mathcal{M}|}$.
\end{lemma}
\begin{proof}
The first and second constraints of \emph{EAP} together with a simple averaging argument implies that, corresponding to each $D \in \mathcal{CM}$, there exists machine $i_D$ such that $\sum\limits_{j \in \mathcal{J}}u_{i_Dj}\mathcal{P}_{12}[i_D,j]\ge \frac{T}{6}$. We form the pair of vectors $\langle \mathbf{u}', \mathbf{s}'\rangle $ as follows.
\begin{itemize}
\item We set $s'_{i_D}=1$ and $u'_{i_Dj}= u_{i_Dj}$ for each $j \in \mathcal{J}$.  
\item We set $s'_{i}=0$ and $u'_{ij}=0$ for each $i \in D$ such that $i \neq i_D$ and $j \in \mathcal{J}$.
\end{itemize}
Now we note the following.
\begin{itemize}
\item $\mathbf{s}' \in \{0,1\}^{|\mathcal{M}|}$ and hence satisfies the set of constraints (7) of \emph{EAP}.
\item For each $i \in \mathcal{M}$ and $j \in \mathcal{J}$, we see that $0\le u'_{ij}\le u_{ij}\le 1$. Hence $\mathbf{u}$ satisfies the set of constraints (6) of \emph{EAP}.
\item For each $D \in \mathcal{CM}$, we see that there exists exactly one machine $i_D$ such that $s'_{i_D}=1$ and for each machine $i \in D$ such that $i \neq i_D$, we have $s'_i=0$. Hence $\sum\limits_{i \in D}s'_i=1$. Hence $\mathbf{s}'$ satisfies the set of constraints (4) of \emph{EAP}.
\item For each $j \in \mathcal{J}$, we see that 
\begin{eqnarray*}
\sum\limits_{i \in \mathcal{M}}u'_{ij} =  \sum\limits_{D \in \mathcal{CM}}\sum\limits_{i \in D} u'_{ij} 
\leq  \sum\limits_{D \in \mathcal{CM}} \sum\limits_{i \in D} u_{ij} 
 \le 1
\end{eqnarray*}
Hence the pair $\langle \mathbf{u},\mathbf{s} \rangle $ satisfies the set of constraints (5) of \emph{EAP}. This proof follows from Lemma \ref{lem:eapfeasible}.
\item For each $D \in \mathcal{CM}$, by the choice of $i_D$ we see that 
\begin{eqnarray*}
\sum\limits_{i \in D}s'_i\sum\limits_{j \in \mathcal{J}}u'_{ij}\mathcal{P}_{12}[i,j] & = &\sum\limits_{j \in \mathcal{J}}u'_{i_Dj}\mathcal{P}_{12}[i_D,j] = \sum\limits_{j \in \mathcal{J}}u_{i_Dj}\mathcal{P}_{12}[i_D,j] \geq \frac{T}{6}
\end{eqnarray*}
\end{itemize}
The polynomial time computability of $\langle \mathbf{u}',\mathbf{s}'\rangle $ easily follows from the construction. Hence the lemma is proved.
\end{proof}

\subsubsection{Finding the $\frac{T}{12}$ solution for clustered instance}
\label{findingTby12solutionforcase2}
We now show that $I$ has a polynomial time computable solution that allocates jobs among the machines in $I$ in such a way that each machine gets a total size of at least $\frac{T}{12}$.
\begin{lemma}
\label{finallemma}
If the given Santa Claus instance $I$ has one or more upper class machines with respect to the solution $\mathbf{x}$ of $CLP(I_{12},T)$, then a solution which allocates job configurations of size at least $\frac{T}{12}$ to each machine in $I$ can be found in polynomial time.
\end{lemma}
\begin{proof}
Let the given Santa Claus instance $I$ has one or more upper class machines with respect to the parameter $T$. We now form the composite machines and solve the corresponding \emph{ECLP} to obtain a feasible solution $\mathbf{z}$. Then we compute the pair $\langle \mathbf{u},\mathbf{s}\rangle$ representing the solution of the corresponding \emph{EAP} using the approach  given in Section \ref{computingTby6solution}. Then obtain the \emph{EAP} solution $\langle \mathbf{u}',\mathbf{s}'\rangle $ as described in Lemma \ref{roundinglemma}. Now consider the instance $I'$ consisting of machines in the set $\Big \{i_D \Big | D \in \mathcal{CM} \Big\}$ and the small jobs in $\mathcal{J}_{small}$. Clearly the solution $\mathbf{v}=\Big[u'_{i_Dj} \Big | j \in \mathcal{J}_{small} \Big]$ is a feasible solution for $ALP(I')$ with objective function at least $\frac{T}{6}$. Now applying Lemma \ref{bezakovadani} on $\mathbf{v}$, we get a $\{0,1\}$ solution $\mathbf{v}'$ that ensures that each machine gets small job configurations with size of at least $\frac{T}{6}-\frac{T}{12}=\frac{T}{12}$. Now we form the solution for $I$ as follows.
\begin{itemize}
\item For each composite machine $D \in \mathcal{CM}$,
\begin{itemize}
\item If $D$ corresponds to a middle class machine $i$, then clearly $i_D=i$ and we allocate to $i$, all the small jobs $j$ such that $v'_{i_Dj}=1$. 
\item If $D$ corresponds to a super machine $M$, then we choose the machine $i_D$ from $M$ and allocate to $i_D$, all the small jobs $j$ such that $v'_{i_Dj}=1$. For each machine $i \in M$ such that $i\neq i_D$, assign a single big gift from the corresponding super job $J$.
\end{itemize}
\end{itemize}
By the way of construction of this solution, we see that corresponding to each composite machine, there is exactly one machine which is allocated a small job configuration of size at least $\frac{T}{12}$ and each of the remaining machines is assigned a single big job which has size at least $\frac{T}{12}$ in $I$. Hence we see that the solution ensures that each machine in $I$ gets job configurations of size at least $\frac{T}{12}$. 

The polynomial time computability of the solution follows from the facts that 
\begin{itemize}
\item \emph{ECLP} can be solved in polynomial time.
\item The vectors $\langle \mathbf{u},\mathbf{s}\rangle $ and $\langle \mathbf{u}',\mathbf{s}'\rangle $ can be found in polynomial time as evident from Lemma \ref{lem:lemma0}and Lemma \ref{roundinglemma}.
\item The integer solution $\mathbf{v}'$ can be obtained from the fractional solution $\mathbf{v}$ in polynomial time as evident from Lemma \ref{bezakovadani}.
\end{itemize}
Hence the proof.
\end{proof}
\subsection{Proving Theorem \ref{thm1}}
We now restate Theorem 1 and give a proof of the same.
\begin{theoremdup}
There exists a polynomial time 12-approximation algorithm for the restricted Santa Claus problem.
\end{theoremdup}
\begin{proof}
The algorithm follows that outline given in Section \ref{sec:algooutline}.
If the given Santa Claus instance $I$ has no upper class machines with respect to solution $\mathbf{x}$ of $CLP(I_{12},T)$, then by Lemma \ref{noupperclasslemma}, we see that there exists a polynomial time computable solution with minimum total size of at least $\frac{T}{12}$. On the other hand, if $I$ has one or more upper class machines with respect to $\mathbf{x}$, then by Lemma \ref{finallemma}, there exists a polynomial time computable solution with minimum total size $\frac{T}{12}$. Hence in either case, we can compute a solution for $I$ with minimum total size $\frac{T}{12}$ in polynomial time. Hence the proof.
\end{proof}

\section{Conclusion}
We have given a linear programming based polynomial time algorithm for the restricted Santa Claus problem with approximation ratio of 12.
We believe that our technique along with the technique of Chidambaram \textit{et.al} \cite{annamalai1} can be used improve the approximation ratio further.  

\bibliographystyle{unsrt}	
\nocite{*}	
\bibliography{boundedpathwidthtreewidth.bib}
\appendix
\end{document}